\documentclass[11pt]{article}

\usepackage{url,amsmath,amsthm,amscd,amsfonts,graphicx,psfrag,marginnote,hyperref}
\usepackage[all]{xy}
\usepackage{color}
\usepackage{fancyhdr}
\usepackage{wrapfig}

\newtheoremstyle{stuff}{\topsep}{\topsep}%
     {\itshape}
     {}
     {\bfseries}
     {}
     {.5em}
     {\thmnote{#3}}

\theoremstyle{plain}
\newtheorem{thm}{Theorem}

\newtheorem{lem}[thm]{Lemma}

\theoremstyle{remark}
\newtheorem*{rem}{Remark}

\theoremstyle{definition}
\newtheorem{defn}[thm]{Definition}

\newtheorem{example}[thm]{Example}

\theoremstyle{stuff}

\numberwithin{equation}{section}
\numberwithin{thm}{section}
 
\usepackage[scale=0.72]{geometry}

\usepackage{sectsty}
\sectionfont{\center \sc \large}
\subsectionfont{\bf \large}
\subsubsectionfont{\it \large}

\usepackage{setspace}

\newcommand{\bea}{\begin{eqnarray}}
\newcommand{\eea}{\end{eqnarray}}
\newcommand{\be}{\begin{equation}}
\newcommand{\ee}{\end{equation}}

\begin{document}

\pagestyle{empty}

\begin{flushright}
\begin{tabular}{l}
 \\
\end{tabular}
\end{flushright}

{\noindent \Large \bf\fontfamily{pag}\selectfont A generalized topological recursion for arbitrary ramification}\\

\vspace*{0.15cm}
\noindent \rule{\linewidth}{0.5mm}

\vspace*{1cm}

{\noindent \fontfamily{pag}\selectfont  Vincent Bouchard, Joel Hutchinson, Prachi Loliencar,\\[0.5em] Michael Meiers and Matthew Rupert}\\[2em]
{\small {\it  \indent Department of Mathematical and Statistical Sciences\\
\indent University of Alberta\\
\indent 632 CAB, Edmonton, Alberta T6G 2G1\\
\indent Canada}\\
\indent \url{vincent@math.ualberta.ca}\\
\indent\url{jhutchin@ualberta.ca}\\
\indent \url{lolienca@ualberta.ca}\\
\indent\url{mcmeiers@gmail.com}\\
\indent \url{mrupert@ualberta.ca}}\\[0.3em]

\vspace*{1cm}

\hspace*{1cm}
\parbox{11.5cm}{{\sc Abstract:} The Eynard-Orantin topological recursion relies on the geometry of a Riemann surface $S$ and two meromorphic functions $x$ and $y$ on $S$. To formulate the recursion, one must assume that $x$ has only simple ramification points. In this paper we propose a generalized topological recursion that is valid for $x$ with arbitrary ramification. We justify our proposal by studying degenerations of Riemann surfaces. We check in various examples that our generalized recursion is compatible with invariance of the free energies under the transformation $(x,y) \mapsto (y,x)$, where either $x$ or $y$ (or both) have higher order ramification, and that it satisfies some of the most important properties of the original recursion. Along the way, we show that invariance under $(x,y) \mapsto (y,x)$ is in fact more subtle than expected; we show that there exists a number of counter examples, already in the case of the original Eynard-Orantin recursion, that deserve further study.}

\vspace*{1cm}


\vspace*{0.5cm}


\pagebreak

\pagestyle{fancy}
\pagenumbering{arabic}

\tableofcontents 


%

\section{Introduction}

The Eynard-Orantin topological recursion was proposed by Eynard and Orantin in \cite{Eynard:2007,Eynard:2008}, based on the original formulation by Chekhov, Eynard and Orantin in the context of two-matrix models \cite{CEO:2006}. In recent years, it has become clear that the recursion is a unifying theme in various counting problems in enumerative geometry. For instance, it appears in Hurwitz theory \cite{Bouchard:2009ii,BEMS,EMS:2009,MZ}, Gromov-Witten theory \cite{Marino:2008,Bouchard:2009, Bouchard:2010, Chen:2009,Zhou:2009,Zhou:2009ii,Zhu:2011,BCMS,Eynard:2011,NS:2011,Eynard:2011ii,EO:2012}, Seiberg-Witten theory \cite{HK}, knot theory \cite{BEM,BE,GS,DFM}, and many more contexts \cite{DMSS,NS:2010}.

The Eynard-Orantin recursion is based on the geometry of a spectral curve $(S,x,y,W^0_2)$, where $S$ is a compact Riemann surface, $x$ and $y$ are meromorphic functions on $S$, and $W^0_2$ is a canonical bilinear differential on $S$ (see section \ref{s:EO} for definitions). Moreover, one needs to make two assumptions on the functions $x$ and $y$:
\begin{enumerate}
\item  the ramification points of $x$ and $y$ must not coincide;
\item  the ramification points of $x$ must all be simple.
\end{enumerate}
While the first assumption is rather important, the second is merely technical, in the sense that the recursion as formulated by Eynard and Orantin only makes sense for simple ramification points, but fundamentally there is no reason to consider only such cases. In this paper we generalize the recursion to get rid of the second assumption on the function $x$, \emph{i.e.} we allow arbitrary ramification.

The case where $x$ has double ramification points was already studied from the point of view of Cauchy matrix models by Prats Ferrer in \cite{Prats:2010}. Our generalization is based on  \cite{Prats:2010}, even though we work outside of the matrix model realm.

The structure of the original Eynard-Orantin recursion can be understood in terms of degenerations of Riemann surfaces. In the same vein, we justify the structure of our generalized recursion by studying more complicated degenerations of Riemann surfaces appropriate for higher order ramification points.

One of the most fascinating properties of the free energies $F_g$ constructed from the Eynard-Orantin recursion is symplectic invariance. In \cite{Eynard:2007ii}, Eynard and Orantin argue that the $F_g$ constructed from a spectral curve $(S,x,y,W^0_2)$ should be equal to the $\tilde{F}_g$ constructed from a distinct spectral curve $(S,\tilde{x},\tilde{y},W^0_2)$, with $(\tilde{x},\tilde{y}) = (y,x)$. This is a highly non-trivial statement, since the ramification points of $x$ and $\tilde{x} = y$ have nothing in common \emph{a priori}. 

Of course, to make sense of this statement using the original Eynard-Orantin recursion, one must require that not only $x$ has only (and at least one) simple ramification points, but that $\tilde{x} = y$ also has only (and at least one) simple ramification points. Using our generalized recursion, we can get rid of this assumption. We check in various examples that in cases where $x$ has only simple ramification points but $y$ has arbitrary ramification, the $F_g$ constructed from the Eynard-Orantin recursion for $(x,y)$ are precisely equal to the $\tilde{F}_g$ constructed from our generalized recursion for $(\tilde{x},\tilde{y}) = (y,x)$. Moreover, we also check that the free energies constructed from the generalized recursion when both $x$ and $y$ have higher order ramification points are invariant under $(x,y) \mapsto (y,x)$. These results support the claim that our generalization is the appropriate one for meromorphic functions $x$ and $y$ with arbitrary ramification.

We also check in examples that the correlation functions and free energies constructed from the generalized recursion satisfy some of the most important properties fulfilled by the corresponding objects in the original Eynard-Orantin recursion, such as symmetry and the dilaton equation. This provides further evidence that our generalization is appropriate.

Along the way, we find that invariance of the free energies under $(x,y) \mapsto (y,x)$ is in fact more subtle than expected. We find a number of counter examples where $F_g \neq \tilde{F}_g$. Some of these examples only involve the original Eynard-Orantin recursion, hence the argument of \cite{Eynard:2007ii} must fail somehow for these spectral curves. We speculate on what is going wrong in these examples, but we cannot identify the precise step in \cite{Eynard:2007ii} that fails for these examples. This shows that symplectic invariance should be studied more carefully, and we hope to report on that in the near future.

\subsection*{Outline}

In section \ref{s:EO} we review the original Eynard-Orantin topological recursion, and study how it relates to degenerations of Riemann surfaces. In section \ref{s:arbitrary} we define our generalized recursion for arbitrary ramification. We argue that it is a natural generalization from the point of view of gluing Riemann surfaces. In section \ref{s:si} we discuss symplectic invariance of the free energies, and check in various examples that our generalized recursion is consistent with symplectic invariance. We also check in \ref{s:other} that the correlation functions are symmetric, and that they satisfy the dilaton equation. In \ref{s:SI} we study more carefully symplectic invariance, and present explicit counter examples that need to be studied further. We conclude with future avenues of research in section \ref{s:conclusion}.

\subsection*{Acknowledgments}

V.B. would like to thank Andrei Catuneanu, Bertrand Eynard, Amir Kashani-Poor, Christopher Marks, Motohico Mulase and Piotr Su\l kowski for useful discussions. The research of V.B. is supported by an NSERC Discovery grant, while J.H, P.L., M.M and M.R. are supported by NSERC Undergraduate Student Research Awards.

\section{The Eynard-Orantin topological recursion}
\label{s:EO}

In this section we review the original topological recursion formulated by Eynard and Orantin in \cite{Eynard:2007, Eynard:2008}. We will use the notation put forward by Prats Ferrer in \cite{Prats:2010}.

\subsection{Geometric context}

We start with a compact Riemann surface $S$. The type of objects that we will be interested in are \emph{meromorphic differentials} on the Cartesian product
\begin{equation}
S^n = \underbrace{S \times \ldots \times S}_{n \text{ times}}.
\end{equation}
In local coordinates $z_i := z(p_i)$, $p_i \in S$, $i=1,\ldots,n$ a degree $n$ differential can be written as
\begin{equation}
W_n(p_1, \ldots, p_n) = w_n(z_1, \ldots, z_n) \mathrm{d} z_1 \cdots \mathrm{d} z_n,
\end{equation}
where $w(z_1, \ldots, z_n)$ is meromorphic in each variable.

We now define one such meromorphic differential, which is a classical object in the theory of Riemann surfaces, and is fundamental for the topological recursion:
\begin{defn}
The \emph{canonical bilinear differential of the second kind} $W^0_2(p_1,p_2)$ is the unique bilinear differential on $S^2$ defined by the conditions:\footnote{Note that in the physics literature the canonical bilinear differential has also been called \emph{Bergman kernel}.}
\begin{itemize}
\item It is symmetric, $W^0_2(p_1, p_2) = W^0_2(p_2,p_1)$;
\item It has its only pole, which is double, along the diagonal $p_1 = p_2$, with no residue; its expansion in this neighborhood has the form
\begin{equation}
W^0_2(p_1, p_2) = \left(\frac{1}{(z_1-z_2)^2} + \text{regular} \right) \mathrm{d} z_1 \mathrm{d} z_2.
\end{equation}
\item It is normalized about a basis of $(A^I, B_J)$ cycles on $S$ such that
\begin{equation}
\oint_{A^I} B( \cdot, p) = 0.
\end{equation}
\end{itemize}
\end{defn}

\begin{example}
For $S = \mathbb{P}^1$, the canonical bilinear differential is the Cauchy differentiation kernel:
\begin{equation}
W_2^0(p_1, p_2) = \frac{\mathrm{d} z_1 \mathrm{d} z_2}{(z_1 - z_2)^2}.
\end{equation}
It is such that for any meromorphic function $f$ on $S$, we have that
\begin{equation}
\mathrm{d} f = f'(z) \mathrm{d} z =  \underset{w = z}{\text{Res}} \left( \frac{f(w) \mathrm{d} z \mathrm{d} w}{(w-z)^2} \right) =   \underset{w = z}{\text{Res}} \left( f(w) W_2^0(w,z) \right).
\end{equation}
\end{example}

\begin{rem}
Note that the property that
\begin{equation}
\mathrm{d} f = \underset{w = z}{\text{Res}} \left( f(w) W_2^0(w,z) \right)
\end{equation}
for any meromorphic function on $S$ is true not only when $S$ has genus $0$. In fact, it could be taken as a definition of the canonical bilinear differential.
\end{rem}

We can now define the key player in the game of the topological recursion:

\begin{defn}\label{d:spectralEO}
A \emph{spectral curve} is a quadruple $(S, x, y, W_2^0)$ where:
\begin{itemize}
\item $S$ is a compact Riemann surface;
\item $x$ and $y$ are meromorphic functions on $S$;\footnote{Note that the recursion can also be formulated for $x$ and $y$ not meromorphic on $S$, as long as their differentials $\mathrm{d} x$ and $\mathrm{d} y$ are meromorphic one-forms on $S$.}
\item $W_2^0(p_1,p_2)$ is a canonical bilinear differential of the second kind.
\end{itemize}
We also assume that $x$ has only simple ramification points, and that the ramification points of $x$ and $y$ do not coincide.
\end{defn}

\begin{rem}
The requirement that $x$ has only \emph{simple} ramification points is necessary to make sense of the original Eynard-Orantin recursion. This is precisely the requirement that we will get rid of by proposing a generalized topological recursion.
\end{rem}

To define the topological recursion in the next subsection, we need to introduce some notation. We assume that the meromorphic function $x$ on $S$ has only simple ramification points. Let $\triangle$ be the set of simple ramification points of $x$. At each $a \in \triangle$, two of the branches of $x$ collide, since the ramification points are simple. This means that there exists a unique involution
\begin{equation}
\theta: U \to U,
\end{equation}
defined locally in a neighborhood $U$ of $a$, that exchanges the two branches; it is given by the deck transformation:
\begin{equation}
x \circ \theta = x,
\end{equation}
with the requirement that $\theta(a) = a$, \emph{i.e.} it exchanges the two branches that meet at $a$. It is easy to show that $\theta(q)$ has a series expansion of the form
\begin{equation}
\theta(q) = a - (q-a) + \mathcal{O}(q-a)^2.
\end{equation}

\subsection{The topological recursion}

Now given a spectral curve $(S, x,y, W_2^0)$, we construct recursively an \emph{infinite tower of meromorphic differentials} on $S^n$. Let $\{ W^g_n \}$ be an infinite sequence of degree $n$ meromorphic differentials $W^g_n(p_1, \ldots, p_n)$ for all integers $g\geq 0$ and $n > 0$ satisfying the condition $2g - 2 + n \geq 0$. We say that the differentials with $2g-2+n > 0$ are \emph{stable}; $W_2^0(p_1, p_2)$ is the only unstable differential, given by the canonical bilinear differential of the second kind. We will call these differentials \emph{correlation functions}.

\begin{defn}\label{d:curlyEO}
Let $\mathbf{p} = \{p_1, \ldots, p_n \}$. With the purpose of generalizing the recursion in the next section, we define meromorphic differentials:
\begin{equation}\label{e:curlyEO}
\mathcal{W}^g_{2,n}(q_1,q_2; \mathbf{p}) = W^{g-1}_{n+2} (q_1, q_2, \mathbf{p}) +\sum'_{\substack{g_1 + g_2 = g \\ \mathbf{r}_1 \cup \mathbf{r}_2 =  \mathbf{p}}} W^{g_1}_{|\mathbf{r}_1|+1}(q_1, \mathbf{r}_1) W^{g_2}_{|\mathbf{r}_2|+1}(q_2, \mathbf{r}_2).
\end{equation}
The sum is over all non-negative integer pairs $(g_1,g_2) \in \mathbb{N}^2$ such that $g_1+g_2 = g$, and all non-intersecting subsets $\mathbf{r}_1, \mathbf{r}_2 \subseteq \mathbf{p}$ such that $\mathbf{r}_1 \cup \mathbf{r}_2 = \mathbf{p}$. Moreover, the prime over the summation symbol means that we exclude the cases $(g_1, \mathbf{r}_1) = (0, \emptyset)$ and $(g_2, \mathbf{r}_2) = (0, \emptyset)$.
\end{defn}

\begin{defn}
We define the \emph{Eynard kernel} $K_2(p_0;q_1,q_2)$ by
\begin{equation}
K_2(p_0;q_1,q_2) = - \frac{dS_{q_1,0}(p_0)}{\omega(q_1,q_2)},
\end{equation}
where $dS_{q_1,0}(p_0)$ is the canonical normalized Abelian differential of the third kind: 
\begin{equation}
dS_{q_1,0}(p_0) = \int_0^{q_1} W^0_2(p_0, \cdot),
\end{equation}
with $0$ an arbitrary base point (it has simple poles at $p_0=q_1$ and $p_0=0$ with respective residues $+1$ and $-1$), and $\omega(q_1,q_2)$ is defined in terms of the functions $x$ and $y$ by
\begin{equation}
\omega(q_1,q_2) = (y(q_1) - y(q_2)) \mathrm{d} x(q_1).
\end{equation}
Note that $1/\mathrm{d} x(q_1)$ here is the contraction operator with respect to the vector field $\frac{1}{\mathrm{d} x/\mathrm{d}q_1} \frac{\partial}{\partial q_1}$.
\end{defn}

We are now ready to define the Eynard-Orantin topological recursion:
\begin{defn}\label{d:recursionEO}
Let $\mathbf{p} = \{p_1, \ldots, p_{n} \}$. Let $(S,x,y,W^0_2)$ be a spectral curve, and denote by $\triangle$ the set of simple ramification points of $x$ and $\theta$ the corresponding deck transformation defined locally near $a \in \triangle$. We say that the meromorphic differentials $W_n^g$ satisfy the \emph{Eynard-Orantin topological recursion}  if \cite{Eynard:2007,Eynard:2008}:
\begin{equation}\label{e:recursionEO}
W^g_{n+1}(p_0, \mathbf{p}) =  \sum_{a \in \triangle} \underset{q=a}{\text{Res}} \left( K_2(p_0;q, \theta(q) ) \mathcal{W}^g_{2,n}(q, \theta(q); \mathbf{p} ) \right).
\end{equation}
From Definition \ref{d:curlyEO}, we see that the recursion is on the integer $2g - 2 + n$, which is why it is called a \emph{topological recursion}. The initial condition of the recursion is given by the canonical bilinear differential $W_2^0$ defined above.\footnote{Note that the recursion kernel in \cite{Eynard:2007,Eynard:2008} differs from ours by the fact that the numerator $dS_{q,0}(p_0)$ is replaced by $\frac{1}{2} dS_{q, \theta(q)}(p_0)$. However, it is straightforward to show that the two formulations are equivalent. It follows from the fact that the term multiplied by the recursion kernel is symmetric under the exchange $q \mapsto \theta(q)$, hence we can replace the recursion kernel by its symmetrization, which amounts to replacing $dS_{q,0}(p_0)$ by $\frac{1}{2} dS_{q, \theta(q)}(p_0)$. We would like to thank Bertrand Eynard for discussions on this topic.}
\end{defn}

The correlation functions $W^g_n$ are well defined meromorphic differentials on $S^n$.  It turns out that they satisfy many fascinating properties. First, they are symmetric differentials, which is far from obvious from the definition.  Also, one can show that they satisfy the \emph{dilaton equation}:
\begin{thm}[\cite{Eynard:2007,Eynard:2008}]
\begin{equation}
W^g_{n} (\mathbf{p})= \frac{1}{2g-2+n} \sum_{a \in \triangle}\underset{q=a}{{\rm Res}} \Phi(q) W_{n+1}^g(q,\mathbf{p}),
\end{equation}
where $\Phi(q)$ is a primitive of $y(t) \mathrm{d} x(t)$
\begin{equation}
\Phi(q) = \int_0^q y(t) \mathrm{d} x(t),
\end{equation}
for an arbitrary base point $0$.
\end{thm}

It is thus natural to define ``level $n=0$'' objects, that are just numbers:
\begin{defn}
The numbers $F_g$, $g \geq 2$, that we call \emph{free energies}, are constructed from the one-forms $W_1^g(q)$ by:
\begin{equation}
F_g = \frac{1}{2g-2} \sum_{a \in \triangle} \underset{q=a}{\text{Res}} \Phi(q) W_1^g(q).
\end{equation}
\end{defn}

To summarize, given a \emph{spectral curve} $(S,x,y,W^0_2)$, the Eynard-Orantin topological recursion constructs recursively an infinite tower of symmetric meromorphic differentials $W_n^g(p_1, \ldots, p_n)$ satisfying the stability condition $2g-2+n>0$. The recursion kernel is the \emph{Eynard kernel}. We can extend the recursion and define auxiliary numbers, the free energies $F_g$. It can be shown that the differentials $W_n^g$ and the free energies $F_g$ satisfy all kinds of nice properties, see \cite{Eynard:2007,Eynard:2008}.

\subsection{Degenerations of Riemann surfaces}

\label{s:degeneration1}

For the purpose of generalizing the recursion in the next section, it is very useful to consider a pictorial representation of the structure of the Eynard-Orantin recursion that stems from degenerations of Riemann surfaces. In fact, this representation is at the foundation of why the recursion has so many applications in enumerative geometry.

What we want to understand are the terms on the right-hand-side of the definition of $\mathcal{W}^g_{2,n}(q_1,q_2; \mathbf{p})$, see Definition \ref{d:curlyEO}. Those arise through degenerations of Riemann surfaces.

Consider a genus $g$ compact Riemann surface. A degeneration of a Riemann surface consists in the Riemann surface being ``pinched'' along a cycle. There are two distinct types of degenerations:
\begin{enumerate}
\item The pinched cycle is homologically trivial, in which case the degeneration results in two disconnected compact Riemann surfaces $S_1$ and $S_2$ of genus $g_1$ and $g_2$ with $g_1 + g_2 = g$, each of which with a marked point $q_1 \in S_1$ and $q_2 \in S_2$; 
\item The pinched cycle is not homologically trivial, in which case the degeneration results in a single Riemann surface $S_1$ of lower genus $g_1 = g-1$, with two marked points $q_1, q_2 \in S_1$.
\end{enumerate}

The right-hand-side of \eqref{e:curlyEO} is precisely a summation over all possible degenerations of Riemann surfaces. More precisely, we associate to a correlation function $W^g_n(p_1,\ldots,p_n)$ a compact Riemann surface of genus $g$ with $n$ marked points, labeled by the entries $\{p_1,\ldots, p_n\}$. 

Now consider \eqref{e:curlyEO} for $\mathcal{W}^g_{2,n}(q_1,q_2; \mathbf{p})$. The right-hand-side of \eqref{e:curlyEO} then corresponds to a summation over all possible degenerations of a genus $g$ compact Riemann surface with $n$ marked points labeled by the entries $\{p_1, \ldots, p_n\}$; the two marked points resulting from the generations are labeled by $q_1$ and $q_2$, which are required not to coincide with the original marked points $p_1$ to $p_n$.

For instance, at genus $g=0$, all cycles are homologically trivial, so all degenerations result in two genus $0$  Riemann surfaces $S_1$ and $S_2$. That is, only the second term in the right-hand-side of \eqref{e:curlyEO} survives. The summation is then over all possible ways of splitting the original marked points $\{ p_1, \ldots, p_n \}$ into $S_1$ and $S_2$, each of which has also an additional marked point $q_1 \in S_1$ and $q_2 \in S_2$.

At higher genus, the same type of degenerations exists, but now $S_1$ and $S_2$ must be of genus $g_1$ and $g_2$ with $g_1+g_2 = g$. There is also the possibility of degenerating along a non-trivial cycle, in which case the result is a single Riemann surface $S_1$ of genus $g_1=g-1$, with all the marked points $\{ p_1, \ldots, p_n \}$ plus the two additional marked points $q_1$ and $q_2$. This is the first term on the right-hand-side of \eqref{e:curlyEO}.

\begin{figure}[t]
\begin{center}
\includegraphics[width=\textwidth]{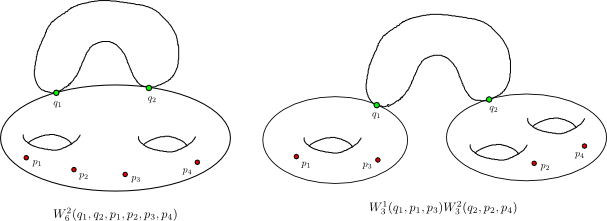} 
\begin{quote}
\caption{ Two degenerations and the corresponding terms entering into the definition of $\mathcal{W}^3_{2,4}(q_1,q_2;p_1,p_2,p_3,p_4)$.} \label{f:degEO}
\end{quote}
\end{center}
\end{figure}

An equivalent, perhaps more constructivist, point of view on degenerations of Riemann surfaces comes from gluing Riemann surfaces. Instead of starting with a genus $g$ compact Riemann surface and looking at all possible degenerations coming from pinching a cycle, one could try to find all possible ways of \emph{constructing} a compact Riemann surface of genus $g$ by gluing a cylinder. The answer would of course be the same. We can either start with two Riemann surfaces $S_1$ and $S_2$ of genus $g_1$ and $g_2$ with $g_1+g_2=g$, and glue a cylinder at two marked points $q_1 \in S_1$ and $q_2 \in S_2$ to obtain a new Riemann surface $S$ of genus $g$; or we can start with a single Riemann surface $S_1$ of genus $g_1=g-1$, and glue a cylinder at two marked points $q_1, q_2 \in S_1$, thus creating one more handle and obtaining a new Riemann surface $S$ of genus $g$. 
This picture will be very useful for generalizing the recursion for arbitrary ramification.

Examples of degenerations and the corresponding terms in \eqref{e:curlyEO} are shown in figure \ref{f:degEO}.

\section{A topological recursion for arbitrary ramification}

\label{s:arbitrary}

In this section we generalize the Eynard-Orantin topological recursion for meromorphic functions $x$ and $y$ with arbitrary ramification. Our generalization is based on the work of Prats Ferrer \cite{Prats:2010}, where the case of double ramification points was studied from the point of view of Cauchy matrix models. Here, we propose our generalization by studying degenerations of Riemann surfaces. We then argue in the next section that the generalization is appropriate by showing in examples that it restores symplectic invariance for cases where $x$ or $y$ (or both) have arbitrary ramification.

\subsection{The geometry}

First, we modify the definition of a spectral curve, Definition \ref{d:spectralEO}, by dropping the assumption on the ramification points of $x$ (we however keep the requirement that the ramification points of $x$ and $y$ do not coincide):
\begin{defn}\label{d:spectral}
A \emph{spectral curve} is a quadruple $(S, x, y, W_2^0)$ where:
\begin{itemize}
\item $S$ is a compact Riemann surface;
\item $x$ and $y$ are meromorphic functions on $S$;
\item $W_2^0(p_1,p_2)$ is a canonical bilinear differential of the second kind.
\end{itemize}
We assume that the ramification points of $x$ and $y$ do not coincide.
\end{defn}

Let us introduce some notation. Recall that the ramification index $e_x(q)$ of $x$ at a point $q \in S$ is given by the number of branches that collide at $q$. All but a finite number of points have $e_x(q)=1$; the points $a \in S$ with $e_x(a) > 1$ are the ramification points.

Assume that $x$ is degree $m$. Let $\triangle$ be the set of ramification points of $x$. We split $\triangle$ into subsets $\triangle_i$, with $i=2,\ldots,m$, according to the ramification indices. In the case studied by Eynard and Orantin, and presented in the previous section, $\triangle = \triangle_2$, that is, it was assumed that all ramification points have $e_x(a) = 2$.

Consider a ramification point $a \in \triangle_k$ with ramification index $e_x(a)=k$. This means that $k$ branches of $x$ collide at $a$. Let $U$ be a neighborhood of $a$. Then there exists $k-1$ distinct non-trivial maps
\begin{equation}
\theta^{(j)}: U \to U, \quad j=1,\ldots,k-1
\end{equation}
that permute the colliding branches; they are the $k-1$ distinct non-trivial deck transformations
\begin{equation}
x \circ \theta^{(j)} = x
\end{equation}
with the requirement that $\theta^{(j)}(a) = a$. With appropriate ordering of the branches, it is clear that
\begin{equation}
\theta^{(m)}( \theta^{(n)} ( q)) = \theta^{( m + n \text{ mod } k)}(q),
\end{equation}
where we defined the trivial map
\begin{equation}
\theta^{(0)}(q) = q.
\end{equation}
It is also straightforward to show that the maps $\theta^{(j)}$, $j=1,\ldots,k-1$, have series expansion of the form
\begin{equation}
\theta^{(j)}(q) = a + \mathrm{e}^{2 \pi i (j/k)} (q-a) + \mathcal{O}(q-a)^2.
\end{equation}

\subsection{A generalized recursion}

We are now in a position to generalize the Eynard-Orantin recursion. Just as before, given a spectral curve $(S, x,y, W_2^0)$, we construct recursively an infinite tower of correlation functions $W^g_n$ on $S^n$. The first objects we need to generalize are the meromorphic differentials in Definition \ref{d:curlyEO}.

Recall that a partition of a set $X$ is a collection of nonempty subsets of $X$ such that every element of $X$ is in exactly one of these subsets. To fix notation, a partition $\mu$ of $X$ is a collection of $p = \ell(\mu)$ subsets $\mu_i \subseteq X$, $i=1,\ldots,p$, such that $\cup_{i=1}^p \mu_i = X$ and $\mu_i \cap \mu_j = \emptyset$ for all $i \neq j$. We denote by $|\mu_i|$ the number of elements in the subset $\mu_i$; clearly $\sum_{i=1}^p |\mu_i| = |X|$ and $|\mu_i| \neq 0$ for all $i=1,\ldots,p$. The trivial partition is of course $\mu = \{ X \}$, which has length $p = \ell(\mu ) = 1$, and $\mu_1 = X$.

\begin{defn}\label{d:curly}
Let $\mathbf{p} = \{ p_1, \ldots, p_n \}$, and $\mathbf{q} = \{q_1, \ldots, q_k \}$. We define meromorphic differentials:
\begin{equation}\label{e:curly}
\mathcal{W}^g_{k,n}(\mathbf{q}; \mathbf{p}) = \sum_{\mu \in \text{Partitions}(\mathbf{q})} \sum'_{\substack{\sum_{i=1}^{\ell(\mu)} g_i = g + \ell(\mu) - k \\ \cup_{i=1}^{\ell(\mu)} \mathbf{r}_i =  \mathbf{p}}} \left( \prod_{i=1}^{\ell(\mu)} W^{g_i}_{|\mu_i| + |\mathbf{r}_i|} (\mu_i, \mathbf{r}_i)  \right).
\end{equation}
The first summation is over partitions $\mu$ of the set $\mathbf{q}$. The second summation involves summing over all possible $\ell(\mu)$-tuples of non-negative integers $(g_1, \ldots, g_{\ell(\mu)})$, where $\ell(\mu)$ is the number of subsets in the partition $\mu$, satisfying $\sum_{i=1}^{\ell(\mu)} g_i = g + \ell(\mu) - k$. It also involves summing over all possible non-intersecting subsets $\mathbf{r}_i \subseteq \mathbf{p}$, for $i=1, \ldots, \ell(\mu)$, such that $\cup_{i=1}^{\ell(\mu)} \mathbf{r}_i =  \mathbf{p}$. Finally, the prime over the second summation symbol means that we exclude the cases with $(g_i, |\mu_i| + |\mathbf{r}_i|) = (0,1)$ for some $i$.
\end{defn}

In the next subsection we will justify this particular definition from degenerations of Riemann surfaces, and then will be study how it restores symplectic invariance for arbitrary ramification. But let us first show that it reduces to Definition \ref{d:curlyEO} when $x$ has only simple ramification points, and to eq.(4.19) in \cite{Prats:2010} when $x$ has simple and double ramification points.

\begin{example}
Assume that $x$ has only simple ramification points, that is, $\triangle = \triangle_2$. Thus, every ramification point has ramification index $k=2$, and so we are interested in computing $\mathcal{W}^g_{2,n}(q_1,q_2; \mathbf{p})$. There are only two partitions of $\{q_1, q_2\}$, namely the trivial partition $\mu = \{ \{q_1, q_2\} \}$ and the partition in two subsets $\mu = \{ \{q_1 \}, \{q_2\} \}$. In the first case, $\ell(\mu)=1$, so the second sum collapses and produces a single term:
\begin{equation}
W^{g-1}_{n+2} (q_1, q_2, \mathbf{p}).
\end{equation}
For the second partition, $\ell(\mu)=2$, and the second summation becomes
\begin{equation}
\sum'_{\substack{g_1 + g_2 = g \\  \mathbf{r}_1 \cup \mathbf{r}_2 =  \mathbf{p}}} W^{g_1}_{|\mathbf{r}_1| + 1}(q_1, \mathbf{r}_1) W^{g_2}_{|\mathbf{r}_2| + 1}(q_2, \mathbf{r}_2). 
\end{equation}
Those are precisely the terms in Definition \ref{d:curlyEO}.
\end{example}

\begin{example}
In the case where $x$ has simple and double ramification points, we are interested in computing both $\mathcal{W}^g_{2,n}(q_1,q_2;\mathbf{p})$ as above and $\mathcal{W}^g_{3,n}(q_1,q_2,q_3; \mathbf{p})$. Let us look at the latter.

There are three types of partitions of $\{q_1,q_2,q_3\}$: the trivial partition $ \{ \{q_1,q_2,q_3 \} \}$; partitions of length two $\{ \{q_1, q_2\}, \{q_3\}\}$, $\{ \{q_1, q_3\}, \{q_2\} \}$ and $\{ \{q_2, q_3\}, \{q_1\} \}$; and a partition of length three $\{ \{q_1\}, \{q_2\}, \{q_3\} \}$.

For the partition of length one, the second summation collapses, and we get the term
\begin{equation}
W^{g-2}_{n+ 3} (q_1, q_2, q_3, \mathbf{p}).
\end{equation}
For the partitions of length two, we get the three terms
\begin{gather}
\sum'_{\substack{g_1 + g_2 = g-1 \\  \mathbf{r}_1 \cup \mathbf{r}_2 =  \mathbf{p}}} W^{g_1}_{|\mathbf{r}_1| + 2}(q_1, q_2, \mathbf{r}_1) W^{g_2}_{|\mathbf{r}_2| + 1}(q_3, \mathbf{r}_2),\\
\sum'_{\substack{g_1 + g_2 = g-1 \\  \mathbf{r}_1 \cup \mathbf{r}_2 =  \mathbf{p}}} W^{g_1}_{|\mathbf{r}_1| + 2}(q_1, q_3, \mathbf{r}_1) W^{g_2}_{|\mathbf{r}_2| + 1}(q_2, \mathbf{r}_2),\\
\sum'_{\substack{g_1 + g_2 = g-1 \\  \mathbf{r}_1 \cup \mathbf{r}_2 =  \mathbf{p}}} W^{g_1}_{|\mathbf{r}_1| + 2}(q_2, q_3, \mathbf{r}_1) W^{g_2}_{|\mathbf{r}_2| + 1}(q_1, \mathbf{r}_2).
\end{gather}
As for the partition of length three, we get the term
\begin{equation}
\sum'_{\substack{g_1 + g_2+g_3 = g \\  \mathbf{r}_1 \cup \mathbf{r}_2 \cup \mathbf{r}_3 =  \mathbf{p}}} W^{g_1}_{|\mathbf{r}_1| + 1}(q_1, \mathbf{r}_1) W^{g_2}_{|\mathbf{r}_2| + 1}(q_2, \mathbf{r}_2)W^{g_3}_{|\mathbf{r}_3| + 1}(q_3, \mathbf{r}_3). 
\end{equation}
Those are precisely the terms obtained in eq.(4.19) of \cite{Prats:2010}.
\end{example}

Now that we understand the curly $\mathcal{W}$'s, we also need to generalize the Eynard kernel. The generalization is straightforward from \cite{Prats:2010}:
\begin{defn}\label{d:kernel}
Let $\mathbf{q} = \{q_1, \ldots, q_k \}$. We define the \emph{generalized Eynard kernel} $K_k(p_0;\mathbf{q})$ by
\begin{equation}
K_k(p_0;\mathbf{q}) = - \frac{dS_{q_1,0}(p_0)}{\prod_{i=2}^{k} \omega(q_1, q_i)},
\end{equation}
where as before 
\begin{equation}
dS_{q_1,0}(p_0) = \int_0^{q_1} W^0_2(p_0, \cdot),
\end{equation}
and $\omega(q_1,q_i)$ is defined in terms of the functions $x$ and $y$ by
\begin{equation}
\omega(q_1,q_i) = (y(q_1) - y(q_i)) \mathrm{d} x(q_1).
\end{equation}
\end{defn}

Putting this together, we finally define a generalized recursion:
\begin{defn}\label{d:recursion}
Let $\mathbf{p} = \{ p_1, \ldots, p_n \}$, and $\mathbf{q} = \{q_1, \ldots, q_k \}$. Let $(S,x,y,W^2_0)$ be a spectral curve. Assume that $x$ has degree $m$, and denote by $\triangle$ be the set of ramification points of $x$, with subsets $\triangle_i$ of ramification index $i$, $i=2, \ldots, m$. For each ramification point $a \in \triangle_i$, let $\mathbf{d}= \{\theta^{(1)}(q), \ldots, \theta^{(i-1)}(q) \}$ be the corresponding deck transformations.

We say that the meromorphic differentials $W_n^g$ satisfy the \emph{generalized Eynard-Orantin topological recursion} if:
\begin{equation}\label{e:recursion}
W^g_{n+1}(p_0, \mathbf{p}) =  \sum_{a \in \triangle} \underset{q=a}{\text{Res}} \left( \sum_{\emptyset \neq \mathbf{d}' \subseteq \mathbf{d}}  K_{|\mathbf{d}'|+1}(p_0;q,\mathbf{d}') \mathcal{W}^g_{|\mathbf{d}'|+1,n}(q, \mathbf{d}'; \mathbf{p} ) \right).
\end{equation}
The second summation is over all non-empty subsets of $\mathbf{d}= \{\theta^{(1)}(q), \ldots, \theta^{(i-1)}(q) \}$, where $i$ is the ramification index of the corresponding ramification point $a \in \triangle$. In order words, for a ramification point $a$ of index $i$, all curly $\mathcal{W}^g_{s,n}$ with $2 \leq s \leq i$ participate in the recursion.
\end{defn}

This recursion is a straightforward generalization of eq.(4.18) in \cite{Prats:2010}. Let us show that it reduces to Definition \ref{d:recursionEO} for the case of simple ramification points, and to eq.(4.18) of \cite{Prats:2010} for simple and double ramification points.

\begin{example}
Assume that $x$ has only simple ramification points, that is, $\triangle = \triangle_2$. In this case, for each ramification point $a \in \triangle$, we have that $\mathbf{d} = \{ \theta(q) \}$, since the ramification index is two. Therefore, the second summation in \eqref{e:recursion} collapses and we obtain
\begin{equation}
W^g_{n+1}(p_0, \mathbf{p}) =  \sum_{a \in \triangle} \underset{q=a}{\text{Res}}\left( K_2(p_0;q,\theta(q)) \mathcal{W}^g_{2,n}(q, \theta(q); \mathbf{p} ) \right),
\end{equation}
which is just \eqref{e:recursionEO} of Definition \ref{d:recursionEO}.
\end{example}

\begin{example}
Assume that $x$ has only simple and double ramification points, hence $\triangle = \triangle_2 \cup \triangle_3$. For each $a \in \triangle_2$, we have that $\mathbf{d} = \{ \theta(q) \}$ as above, and they contribute to \eqref{e:recursion} as
\begin{equation}
\sum_{a \in \triangle_2} \underset{q=a}{\text{Res}}\left( K_2(p_0;q,\theta(q) ) \mathcal{W}^g_{2,n}(q, \theta(q); \mathbf{p} ) \right).
\end{equation}
For the ramification points $a \in \triangle_3$, we have that $\mathbf{d} = \{ \theta^{(1)}(q), \theta^{(2)}(q) \}$. Thus the second summation runs over the length one subsets $\{\theta^{(1)}(q)\}$, $\{\theta^{(2)}(q)\}$ and length two subset $\{ \theta^{(1)}(q), \theta^{(2)}(q)\}$. The contribution to the recursion is then
\begin{gather}
\sum_{a \in \triangle_3} \underset{q=a}{\text{Res}} \Big(K_2(p_0;q,\theta^{(1)}(q) ) \mathcal{W}^g_{2,n}(q, \theta^{(1)}(q); \mathbf{p} ) + K_2(p_0;q,\theta^{(2)}(q) ) \mathcal{W}^g_{2,n}(q, \theta^{(2)}(q); \mathbf{p} )\nonumber \\ + K_3(p_0;q,\theta^{(1)}(q), \theta^{(2)}(q)) \mathcal{W}^g_{3,n}(q, \theta^{(1)}(q), \theta^{(2)}(q); \mathbf{p} )  \Big).
\end{gather}
These contributions add up to precisely eq.(4.18) in \cite{Prats:2010}. We see that for double ramification points, both $\mathcal{W}^g_{2,n}$ and $\mathcal{W}^g_{3,n}$ contribute.
\end{example}

Finally, as in \cite{Prats:2010}, we can define free energies just as before:
\begin{defn}
The \emph{free energies} $F_g$, $g \geq 2$ ,are constructed from the one-forms $W_1^g(q)$ by:
\begin{equation}
F_g = \frac{1}{2g-2} \sum_{a \in \triangle} \underset{q=a}{\text{Res}} \Phi(q) W_1^g(q).
\end{equation}
\end{defn}

\subsection{Gluing Riemann surfaces}

The key step in the generalization of the recursion resides in the definition of the curly $\mathcal{W}$'s, namely Definition \ref{d:curly}. We will justify our generalization by gluing Riemann surfaces.

As before, we associate to $W^g_n(p_1,\ldots,p_n)$ a compact Riemann surface of genus $g$ with $n$ marked points, labeled by the entries $\{p_1, \ldots, p_n\}$. We saw in subsection \ref{s:degeneration1} that the right-hand-side of the definition of $\mathcal{W}^g_{2,n}(q_1,q_2; \mathbf{p})$ is simply a summation over all possible degenerations of a genus $g$ compact Riemann surface with $n$ marked points labeled by the entries of $\mathbf{p}$. Equivalently, we can understand these degenerations by looking at all possible ways of obtaining a genus $g$ compact Riemann surface with $n$ marked points by gluing a cylinder.

We now want to understand the right-hand-side of the definition of $\mathcal{W}^g_{k,n}(q_1,\ldots,q_k; \mathbf{p})$ for arbitrary $k \geq 2$, see Definition \ref{d:curly}. The main difference is in the numbers of $q$'s. But for $k=2$, the $q$'s are associated to the marked points where we glue the cylinder. So it seems that the generalization should involve gluing a sphere with $k$ marked points, instead of a cylinder (which is a sphere with two marked points).

Indeed, this is precisely what \eqref{e:curly} does for you. Consider first the case $\mathcal{W}^g_{3,n}(q_1,q_2,q_3; \mathbf{p})$. The right-hand-side of \eqref{e:curly} is a summation over all possible ways of obtaining a genus $g$ compact Riemann surface with $n$ marked points by gluing a sphere with three marked points. Indeed, at genus $g=0$, there is a single possibility; one must start with three Riemann surfaces $S_1$, $S_2$ and $S_3$ with genus $g_1+g_2+g_3=g$, and glue along the marked points $q_1 \in S_1$, $q_2 \in S_2$ and $q_3 \in S_3$. At genus $g=1$, there is a new possibility; one can also start with only two Riemann surfaces $S_1$ and $S_2$, and glue along two marked points on one of the two (say $q_1, q_2 \in S_1$) and one marked point on the other one (say $q_3 \in S_2$). The gluing along $S_1$ creates a handle, hence increases the genus by one. Thus we must have $g_1 + g_2 = g-1$. As for higher genus $g \geq 2$, there is a third possibility, where we start with a single Riemann surface $S_1$ and glue along three marked points $q_1,q_2,q_3 \in S_1$. This creates two new handles, hence we must have that $g_1 = g-2$. Those are precisely the terms in the right-hand-side of \eqref{e:curly}.

\begin{figure}[t]
\begin{center}
\includegraphics[width=\textwidth]{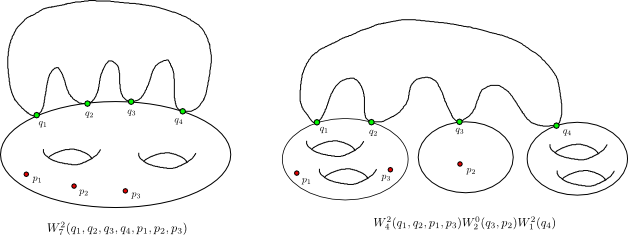} 
\begin{quote}
\caption{ Two degenerations and the corresponding terms entering into the definition of $\mathcal{W}^5_{4,3}(q_1,q_2,q_3,q_4;p_1,p_2,p_3)$.} \label{f:deg}
\end{quote}
\end{center}
\end{figure}

A similar analysis leads to \eqref{e:curly} for higher $k$. The possible gluings are associated to set partitions $\mu$ of $\mathbf{q}$, and the number of Riemann surfaces that one must start with is given by the length $\ell(\mu)$ of the partition $\mu$. For each starting Riemann surface $S_i$, there are $|\mu_i|-1$ handles created in the gluing process. The total number of handles created is $\sum_{i=1}^{\ell(\mu)} |\mu_i| - \ell(\mu) = k - \ell(\mu)$. Thus we must require that the starting Riemann surfaces have genus $\sum_{i=1}^{\ell(\mu)} g_i = g +\ell(\mu)-k$. Note that we must also sum over all possible ways of splitting the marked points $\{ p_1, \ldots, p_n \}$ over the starting Riemann surfaces $S_i$, hence the summation over non-intersecting subsets $\mathbf{r}_i$ of $\mathbf{p}$ with $\cup_{i=1}^{|\mu|} \mathbf{r}_i = \mathbf{p}$.

Examples of degenerations are shown in figure \ref{f:deg}.

\section{Symplectic invariance}
\label{s:si}
\subsection{Definition}

A nice property satisfied by the free energies $F_g$ constructed by Eynard and Orantin is known as \emph{symplectic invariance}. Consider two distinct spectral curves $(S,x,y,W^0_2)$ and $(S,\tilde{x},\tilde{y},W^0_2)$ that differ only by the choice of meromorphic functions $(x,y)$ and $(\tilde{x},\tilde{y})$. Suppose that the meromorphic functions are such that
\begin{equation}
|\mathrm{d} x \wedge \mathrm{d} y| = |\mathrm{d} \tilde{x} \wedge \mathrm{d} \tilde{y}|,
\end{equation}
as symplectic forms on $\mathbb{C}^2$.
Then the claim is that the free energies $F_g$ and $\tilde{F}_g$ cosntructed from the two respective spectral curves are actually equal:
\begin{equation}
F_g = \tilde{F}_g.
\end{equation}
This is of course a highly non-trivial statement, since the ramification points of $x$ and $\tilde{x}$ may be completely different. 

In fact, the non-trivial part of the above statement can be summarized into the statement that for the particular choice $(\tilde{x},\tilde{y}) = (y, x)$, we have $F_g= \tilde{F}_g$. Invariance under this transformation was argued in \cite{Eynard:2007ii}.

\begin{rem}
We note here however that invariance under this transformation is in fact more subtle than expected, since we found a number of counter examples: see subsection \ref{s:SI}. However the examples that we will look at in the next subsection do not share the pathologies of the counter examples presented later on, hence for those cases symplectic invariance is expected to hold.
\end{rem}

To make sense of invariance under $(x,y) \mapsto (y,x)$ using the original Eynard-Orantin recursion we must impose a further assumption on the functions $x$ and $y$. Indeed, for the recursion to be well defined, the function $x$ must have only (and at least one) simple ramification points, and the ramification points of $x$ and $y$ must not coincide. After the transformation, we must of course require the same assumptions on $\tilde{x}$ and $\tilde{y}$ for the recursion to be well defined. In other words, to make sense of the transformation $(\tilde{x},\tilde{y}) = (y,x)$, we must require that not only $x$ has only (and at least one) simple ramification points, but $y$ as well. 

One of the reasons for generalizing the recursion to arbitrary ramification is to make sense of such symplectic transformations for general choices of $x$ and $y$. For instance, in the context of the original Eynard-Orantin topological recursion, we can compute the free energies $F_g$ for a choice of $(x,y)$ with $x$ having only simple ramification points and $y$ having arbitrary ramification. But then, we cannot make sense of the transformation $(\tilde{x},\tilde{y}) = (y,x)$, since the Eynard-Orantin recursion is not well defined for the new functions $(\tilde{x},\tilde{y})$. We need a generalized recursion for arbitrary ramification.

\subsection{Testing our proposal}

In the previous section we proposed a generalized recursion, and justified the key step in the generalization, namely Definition \ref{d:curly}, through a careful analysis of gluing of Riemann surfaces. As for the other steps, namely the definition of the generalized kernel (Definition \ref{d:kernel}) and the generalized recursion (Definition \ref{d:recursion}), they are rather straightforward generalization of the work of Prats Ferrer for double ramification points \cite{Prats:2010}. However, it remains to be shown that the generalized recursion that we propose is appropriate.

One approach would be to try to obtain the generalized recursion as a solution of loop equations for matrix models whose spectral curves have arbitrary ramification, in the spirit of \cite{Prats:2010}. This would be very interesting indeed. However, in this paper we will take a more ``experimental'' approach.

Our main check for the validity of the generalized recursion will involve symplectic invariance, as described in the previous subsection. We will consider situations where $x$ has only simple ramification points and $y$ has ramification points of higher order. Then we can use the standard Eynard-Orantin recursion to compute the $F_g$'s for the pair $(x,y)$. However, to compute the $\tilde{F}_g$ for the pair $(\tilde{x}, \tilde{y}) = (y,x)$, we need to use our generalized recursion. A rather stringent check of our proposal is that the $\tilde{F}_g$ thus computed should be precisely equal to the $F_g$ computed from the Eynard-Orantin recursion. This is what we show next in various examples. We will stick to genus $0$ examples to simplify calculations.

\subsubsection{First example}
As a first example, consider the spectral curve $(S,x,y,W^0_2)$, where $S$ is the Riemann sphere, $W^0_2$ is the canonical bilinear differential on $S$, and the two meromorphic functions are given in parametric form by
\begin{equation}
x = t + \frac{1}{t}, \qquad y = \frac{1}{3} t^3.
\end{equation}
$x$ has two simple ramification points at $a_{\pm}=\pm1$. At each ramification, the deck transformation is simply $\theta(t) = 1/t$.

We can then compute correlation functions and free energies. A few examples of correlation functions:
{\small
\begin{align}
W^0_3(p_0,p_1,p_2) =& \frac{1}{2} \left(\frac{1}{\left(p_0-1\right){}^2 \left(p_1-1\right){}^2
   \left(p_2-1\right){}^2}-\frac{1}{\left(p_0+1\right){}^2 \left(p_1+1\right){}^2
   \left(p_2+1\right){}^2}\right) \mathrm{d} p_0 \mathrm{d} p_1 \mathrm{d} p_2,\nonumber\\
W^1_1(p_0) =&-\frac{p_0 \left(p_0^4-5 p_0^2+1\right)}{3 \left(p_0^2-1\right){}^4} \mathrm{d} p_0,\nonumber\\
W^2_1(p_0) =& \frac{p_0 \left(13 p_0^{16}-143 p_0^{14}+733 p_0^{12}-2273 p_0^{10}+4285 p_0^8-2273 p_0^6+733 p_0^4-143
   p_0^2+13\right)}{9 \left(p_0^2-1\right){}^{10}}\mathrm{d} p_0,\nonumber\\
W^3_1(p_0)=& \frac{1}{81 \left(p_0^2-1\right){}^{16}}\Big(-5482 p_0^{29}+93194 p_0^{27}-748954 p_0^{25}+3769564 p_0^{23}-13231247 p_0^{21}\nonumber\\ &+33949219
   p_0^{19}
-64093893 p_0^{17}+ 84589248 p_0^{15}-64093893 p_0^{13}+33949219 p_0^{11}\nonumber\\&-13231247
   p_0^9
+3769564 p_0^7-748954 p_0^5+93194 p_0^3-5482 p_0 \Big) \mathrm{d} p_0,\nonumber
\end{align}
}
and so on. The resulting free energies for $g=2,3$ are:
\begin{equation}
F_2 = - \frac{13}{72}, \qquad F_3 = \frac{2741}{648} .
\end{equation}

We can now test our recursion by considering the spectral curve $(C, \tilde{x}, \tilde{y}, W^0_2)$, where
\begin{equation}
\tilde{x} = y = \frac{1}{3} t^3, \qquad \tilde{y} = x = t + \frac{1}{t}.
\end{equation}
$\tilde{x}$ has a single ramification point at $a=0$, but it has ramification index three. The two deck transformations at this point are simply $\theta^{(1)}(t) = \mathrm{e}^{2 \pi i/3} t$ and $\theta^{(2)}(t) = \mathrm{e}^{4 \pi i/3} t$.

We can compute correlation functions and free energies. Some examples:
\begin{align}
\tilde{W}^0_3(p_0,p_1,p_2) =& 0,\\
\tilde{W}^1_1(p_0) =& \frac{1}{3 p_0^3} \mathrm{d} p_0,\\
\tilde{W}^2_1(p_0) =& -\frac{13 \left(p_0^2-1\right)}{9 p_0^5} \mathrm{d} p_0,\\
\tilde{W}^3_1(p_0) =& \frac{5482 p_0^8-5482 p_0^6+3100 p_0^2-1225}{81 p_0^{11}} \mathrm{d} p_0.
\end{align}
Note that these are very different from the correlation functions computed from $(x,y)$. However, we obtain the free energies
\begin{equation}
\tilde{F}_2 = - \frac{13}{72}, \qquad \tilde{F}_3 = \frac{2741}{648} ,
\end{equation}
which are precisely the same as $F_2$ and $F_3$! This is a highly non-trivial check that our generalized recursion is appropriate, since it is consistent with symplectic invariance.

\subsubsection{Second example}

The previous example can be generalized to spectral curves $(S,x,y,W^0_2)$, where $S$ is the Riemann sphere, $W^0_2$ is the canonical bilinear differential on $S$, and
\begin{equation}
x = t + \frac{1}{t}, \qquad y = \frac{1}{m} t^m
\end{equation}
for some integer $m \geq 2$. The case with $(x,y)$ can be treated just as above. As for $(\tilde{x},\tilde{y}) = (y,x)$, we note that $\tilde{x}$ has a single ramification point at $a=0$ with ramification index $m$. The deck transformations are given by $\theta^{(k)}(t) = \mathrm{e}^{2 \pi i k /m} t$, for $k=1,\ldots,m-1$. 

We checked that $F_2 = \tilde{F}_2$ and $F_3 = \tilde{F}_3$ for a number of choices of $m$. For instance, for $m=5$, we obtain
\begin{equation}
F_2 = \tilde{F}_2 = - \frac{23}{8} , \qquad F_3 = \tilde{F}_3 = \frac{53459}{40} .
\end{equation}

\subsubsection{Third example}

Consider now a slightly more complicated example. We choose a spectral curve $(S,x,y,W^0_2)$, with $S$ the Riemann sphere, $W_2^0$ the canonical bilinear differential, and
\begin{equation}
x = t + \frac{1}{t}, \qquad y = t^5 + t^4.
\end{equation}
The standard Eynard-Orantin recursion computes correlation functions and free energies. For instance, we get
\begin{align}
W^1_1(p_0) =& \frac{1}{864} \left(\frac{-23 p_0^2+52 p_0-23}{\left(p_0-1\right){}^4}+\frac{27 \left(21 p_0^2+44
   p_0+21\right)}{\left(p_0+1\right){}^4}\right) \mathrm{d} p_0 ,\nonumber\\
W^2_1(p_0) =& \frac{1}{6561 \left(p_0^2-1\right){}^{10}} \Big(-11206282 p_0^{18}+21073570 p_0^{17}+83652574 p_0^{16}-178408841 p_0^{15}\nonumber\\
&-252507089
   p_0^{14}+653452306 p_0^{13}+368487790 p_0^{12}-1346903405 p_0^{11}-188598983 p_0^{10}\nonumber\\
&+1701917665
   p_0^9-188598983 p_0^8-1346903405 p_0^7+368487790 p_0^6+653452306 p_0^5\nonumber\\
&-252507089 p_0^4-178408841
   p_0^3+83652574 p_0^2+21073570 p_0-11206282\Big) \mathrm{d} p_0.\nonumber
\end{align}
The genus $2$ free energy is
\begin{equation}
F_2 = - \frac{759919}{29160}.
\end{equation}

Now consider $(\tilde{x}, \tilde{y}) = (y,x)$. Then $\tilde{x}$ has two ramification points, at $a_1 = -4/5$ and $a_2=0$. The first ramification point $a_1=-4/5$ is simple, and the corresponding deck transformation can be obtained as a power series:
\begin{equation}
\theta(t) =-\frac{4}{5}-\left(t+\frac{4}{5}\right)+\frac{5}{2} \left(t+\frac{4}{5}\right)^2-\frac{25}{4}
   \left(t+\frac{4}{5}\right)^3+\frac{325}{16} \left(t+\frac{4}{5}\right)^4+\mathcal{O}\left(t+\frac{4}{5}\right)^5 .
\end{equation}
The second ramification point $a_2=0$ has ramification index $4$, and the three deck transformations are:
\begin{align}
\theta^{(1)}(t) =& i t+\left(\frac{1}{4}+\frac{i}{4}\right) t^2+\left(\frac{1}{8}-\frac{5 i}{16}\right)
   t^3-\left(\frac{9}{32}+\frac{7 i}{64}\right) t^4-\left(\frac{15}{64}-\frac{157 i}{512}\right)
   t^5+ \mathcal{O}\left(t\right)^6, \nonumber\\
\theta^{(2)}(t)=&-t-\frac{t^2}{2}-\frac{t^3}{4}-\frac{7 t^4}{16}-\frac{17 t^5}{32}+\mathcal{O}\left(t\right)^6,\\
\theta^{(3)}(t)=&-i t+\left(\frac{1}{4}-\frac{i}{4}\right) t^2+\left(\frac{1}{8}+\frac{5 i}{16}\right)
   t^3-\left(\frac{9}{32}-\frac{7 i}{64}\right) t^4-\left(\frac{15}{64}+\frac{157 i}{512}\right)
   t^5+\mathcal{O}\left(t\right)^6 .\nonumber
\end{align}
We then compute correlation functions and free energies. For instance, we get
\begin{align}
\tilde{W}^1_1(p_0) =&\left(\frac{6 p_0^2-5 p_0+10}{64 p_0^4}-\frac{15625 \left(50 p_0^2+65 p_0+14\right)}{1728 \left(5
   p_0+4\right){}^4} \right) \mathrm{d} p_0 ,\\
\tilde{W}^2_1(p_0) =& \frac{1}{6561 p_0^{10} \left(5 p_0+4\right){}^{10}} \Big(109436347656250 p_0^{18}+669694199218750 p_0^{17}\nonumber\\
&+1782837964843750 p_0^{16}+2700686123828125
   p_0^{15}+2549672690546875 p_0^{14}\nonumber\\
&+1537257561831250 p_0^{13}+577436802231250
   p_0^{12}+122815706095000 p_0^{11}\nonumber\\
&+10973538955375 p_0^{10}+157346161650 p_0^9+504462460932
   p_0^8+417749813505 p_0^7\nonumber\\
&+215322815646 p_0^6+14772112641 p_0^5-78646163166 p_0^4-68304314016
   p_0^3\nonumber\\
&-28892859408 p_0^2-6481218240 p_0-617258880\Big) \mathrm{d} p_0.
\end{align}
Note again that the correlation functions $W^2_1(p_0)$ and $\tilde{W}^2_1(p_0)$ are very different. However, computing the genus $2$ free energy, we obtain
\begin{equation}
\tilde{F}_2 = - \frac{759919}{29160},
\end{equation}
which is of course precisely equal to $F_2$!

\subsubsection{Fourth example}

As a last example, we will study a case where both $x$ and $y$ have higher order ramification points. Consider the spectral curve $(S,x,y,W^0_2)$, with $S$ the Riemann sphere, $W_2^0$ the canonical bilinear differential on $S$, and
\begin{equation}
x= (t-1)^4, \qquad y = t^5.
\end{equation}
$x$ has a single ramification point at $a=1$ with ramification index $4$. The deck transformations are 
\begin{equation}
\theta^{(1)}(t) = 1 + i(t-1), \qquad \theta^{(2)}(t)=1 - (t-1), \qquad \theta^{(3)}(t) = 1 - i(t-1).
\end{equation}

We compute the correlation functions and free energies. Some examples:
\begin{align}
W^1_1(p_0) =&\frac{-77 p_0^4+372 p_0^3-680 p_0^2+552 p_0-162}{160 \left(p_0-1\right){}^6} \mathrm{d} p_0 \nonumber\\
W^2_1(p_0) =&\frac{1}{5120000
   \left(p_0-1\right){}^{16}} \Big(1851787360 p_0^{14}-27481058368 p_0^{13}+189393549204 p_0^{12}\nonumber\\
&-803250636784 p_0^{11}+2341503317728
   p_0^{10}-4960971974112 p_0^9+7874782515696 p_0^8\nonumber\\
&-9508486075776 p_0^7+8770192346592 p_0^6-6144199146560
   p_0^5+3215343372105 p_0^4\nonumber\\
&-1217483144740 p_0^3+314954680160 p_0^2-49771491320 p_0+3621948420 \Big)\mathrm{d} p_0\nonumber
\end{align}
The genus $2$ free energy is
\begin{equation}
F_2 = -\frac{815863}{96000}.
\end{equation}

Now consider the case with $(\tilde{x},\tilde{y}) = (y,x)$. $\tilde{x}$ has again a single ramification point, this time at $a=0$, with ramification index $5$ and deck transformations
\begin{equation}
\theta^{(j)}(t) = \mathrm{e}^{2 \pi i j/5} t, \qquad j=1,\ldots,4.
\end{equation}
With this data we can compute correlation functions and free energies. We find:
\begin{align}
\tilde{W}^1_1(p_0) =&\frac{77 p_0^5+90 p_0^4+65 p_0^3+28 p_0^2-8}{160 p_0^7} \mathrm{d} p_0\nonumber\\
\tilde{W}^2_1(p_0) =& \frac{1}{25600
   00 p_0^{19}}\Big( -925893680 p_0^{17}-1073769696 p_0^{16}-769848138 p_0^{15}-340349096 p_0^{14}\nonumber\\
&+170173776
   p_0^{12}+190197280 p_0^{11}+128055552 p_0^{10}+51733632 p_0^9-19702617 p_0^7\nonumber\\
&-18215652 p_0^6-9710740
   p_0^5-2928996 p_0^4+453376 p_0^2+194208 p_0+27456\Big) \mathrm{d} p_0\nonumber
\end{align}
with genus $2$ free energy
\begin{equation}
\tilde{F}_2 = -\frac{815863}{96000},
\end{equation}
again precisely equal to $F_2$!

\subsection{Other properties}
\label{s:other}

In the previous subsection we showed in examples that our generalized recursion is consistent with invariance of the free energies under the transformation $(x,y) \mapsto (y,x)$. This is a highly non-trivial check that our generalized recursion is the appropriate one.

As further checks, we can study other important properties satisfied by the correlation functions constructed from the original Eynard-Orantin recursion.

Some of the important properties are:
\begin{itemize}
\item {\bf Symmetry:} The correlation functions $W^g_{n+1}(p_0, \ldots, p_n)$ are symmetric under permutations of the arguments. Symmetry under permutations involving $p_0$ is highly non-trivial from the definition of the recursion.
\item {\bf Dilaton equation:} The correlation functions satisfy the property:
\begin{equation}
W^g_{n} (\mathbf{p})= \frac{1}{2g-2+n} \sum_{a \in \triangle}\underset{q=a}{{\rm Res}} \Phi(q) W_{n+1}^g(q,\mathbf{p}).
\end{equation}
\item {\bf A particular formula:} The first non-trivial correlation function, $W^0_3(p_0,p_1,p_2)$ satisfies the nice formula:
\begin{equation}
W^0_3(p_0,p_1,p_2) = \sum_{a \in \triangle}\underset{q=a}{{\rm Res}} \frac{W^0_2(q,p_0) W^0_2(q,p_1) W^0_2(q,p_2)}{\mathrm{d} x(q) \mathrm{d} y(q)}.
\end{equation}
\end{itemize}

These properties were proved for the original Eynard-Orantin recursion in \cite{Eynard:2007,Eynard:2008}. As for the generalized recursion, we checked that they are satisfied in all the examples presented in the previous subsection, which is again a highly non-trivial check that our generalization is appropriate. It would however be very nice to provide explicit proofs of these properties for the generalized recursion.

\subsection{Symplectic invariance?}

\label{s:SI}

Symplectic invariance of the free energies is one of their most fascinating properties. However, as we mentioned earlier, it turns out that it is more subtle than expected. In this subsection we study this question in more details. We note however that the problems raised in this subsection do not apply to the examples studied in the previous subsections.

In \cite{BS:2011}, we already noticed that the free energies were not always symplectic invariant, in the case of spectral curves where $x$ and $y$ are not meromorphic on $S$ (they have log singularities). One could then suggest that this is due to the fact that $x$ and $y$ are not meromorphic, since strictly speaking in \cite{Eynard:2007ii} invariance under the transformation $(x,y) \mapsto (y,x)$ has only been argued for meromorphic functions $x$ and $y$.

However, this is not the end of the story. It turns out that it is rather easy to find counter examples of symplectic invariance even with $x$ and $y$ meromorphic. Here is one:
\begin{lem}
Consider the spectral curve $(S,x,y,W^0_2)$, where $S$ is the Riemann sphere, $W^0_2$ is its canonical bilinear differential, and $x$ and $y$ are given in parametric from by:
\begin{equation}
x(t) = t + \frac{1}{t}, \qquad y(t) = (t-b)^2,
\end{equation}
for some constant $b \neq 0,\pm 1$.\footnote{We need to exclude the cases $b=\pm1$ since for those cases the zeroes of $\mathrm{d} x$ and $\mathrm{d} y$ coincide; as for $b=0$, we will discuss this case in the remark below.}
Consider another spectral curve $(S,\tilde{x},\tilde{y},W^0_2)$ with $(\tilde{x},\tilde{y}) = (y,x)$. Then
\begin{equation}
\tilde{F}_2(b) \neq F_2(b),
\end{equation}
for all $b \neq 0, \pm 1$,
and symplectic invariance does not hold.
\end{lem}

\begin{proof}
This is a straightforward calculation. For the case $(S,x,y,W^0_2)$, $x$ has two ramification points at $a= \pm1$, both with deck transformation $\theta(t) = 1/t$. After computing the correlation functions, we obtain the free energy
\begin{equation}
F_2(b) = \frac{45 + 255 b^2 - 80 b^4 + 40 b^6 - 8 b^8}{7680 (b^2-1)^5}.
\end{equation}

Consider now the spectral curve $(S, \tilde{x}, \tilde{y}, W^0_2)$. $\tilde{x}$ has a single ramification point at $a=b$. The deck transformation is $\theta(t) = 2 b-t$. After computing the correlation functions, we obtain
\begin{equation}
\tilde{F}_2(b) = \frac{-8 + 85 b^2 + 175 b^4}{7680 b^2 (b^2-1)^5}.
\end{equation}
The difference is
\begin{equation}
F_2 (b)- \tilde{F}_2 (b)= - \frac{1}{960 b^2},
\end{equation}
which is clearly not zero.
\end{proof}

\begin{rem}
It is interesting to study what happens for $b=0$. $F_2(b)$ has a well defined limit:
\begin{equation}
\lim_{b =0} F_2(b) = - \frac{3}{512},
\end{equation}
which is also what one obtains starting from the curve $x=t+1/t$ and $y=t^2$. As for $\tilde{F}_2(b)$, at $b=0$ it blows up to infinity. However, if we do the calculation for $\tilde{F}_2$ starting with $\tilde{x} = t^2$ and $\tilde{y} = t + 1/t$ directly, then we obtain a finite answer, namely
\begin{equation}
\tilde{F}_2 = - \frac{3}{512} = F_2.
\end{equation}
Thus, as a function of $b$, one could say that $\tilde{F}_2(b)$ is discontinuous at $b=0$; it is precisely such that symplectic invariance is restored at $b=0$.
\end{rem}

In fact, there are many other such counter examples. It seems that for most cases with functions
\begin{equation}
x(t) = t+\frac{1}{t}, \qquad y(t) = P_n(t),
\end{equation}
where $P_n(t)$ is a generic degree $n$ polynomial such that $\mathrm{d} y(t)$ has simple zeroes, symplectic invariance does not hold. Notable exceptions however are when the polynomial $P_n(t)$ has a factor of $t^2$, that is $P_n(t) = t^2 Q_{n-2}(t)$ for some degree $n-2$ polynomial $Q_{n-2}(t)$, in which case it seems that $\tilde{F}_2 = F_2$, just as in the remark above.

We also obtain similar statements with the generalized recursion. While $\tilde{F}_2 = F_2$ holds for curves of the form
\begin{equation}
x(t) = t+\frac{1}{t}, \qquad y(t) = \frac{1}{n} t^n,
\end{equation}
as explained in the previous subsection, it does not seem to hold for curves of the form
\begin{equation}
x(t) = t+\frac{1}{t}, \qquad y(t) = P_n(t), 
\end{equation}
where $P_n(t)$ is a degree $n$ polynomial (giving up the requirement that $\mathrm{d} y(t)$ has only simple zeroes). For instance, as a particular example, if we choose
\begin{equation}
x(t)=t+\frac{1}{t}, \qquad y(t) = (t-2)^4,
\end{equation}
we obtain
\begin{equation}
F_2= -\frac{343}{22674816}
\end{equation}
using the standard Eynard-Orantin recursion, while we get
\begin{equation}
\tilde{F}_2 = -\frac{160471}{14511882240}
\end{equation}
using our generalized recursion.

However, notable exceptions again are when the polynomial $P_n(t)$ has a factor of $t^m$ for some $m \geq 2$, that is $P_n(t) = t^m Q_{n-m}(t)$ for some degree $n-m$ polynomial $Q_{n-m}(t)$. In these cases, symplectic invariance seems to hold.

An obvious question then is why symplectic invariance fails for these particular curves, and what assumptions need to be made on the functions $x$ and $y$ for symplectic invariance to hold. What goes wrong in the argument of \cite{Eynard:2007ii}? It is not clear to us what the precise answer is.

It is interesting to note however that the counter examples above share the property that the one-form $y(t) \mathrm{d} x(t)$ has a pole at $t=0$ with a non-zero residue, hence $\Phi = \int y(t) \mathrm{d} x(t)$, which enters into the calculation of the $F_g$, is not meromorphic over $S$ (it has log singularities). Perhaps this causes a problem in the argument of \cite{Eynard:2007ii} (indeed, in the cases above where symplectic invariance does hold, the pole of $\mathrm{d} x(t)$ at $t=0$ is cancelled by a zero of $y(t)$ so that $y(t) \mathrm{d} x(t)$ only has a pole at infinity with no residue.) For instance, the Riemann bilinear identity, presented in the form:
\begin{equation}
\underset{q \to \alpha}{\text{Res}} \Phi(q) W^g_1(q) = \frac{1}{2 \pi i} \sum_{I=1}^g \left( \oint_{B_I} y \mathrm{d} x \oint_{A^I} W^g_1 - \oint_{A^I} y \mathrm{d} x \oint_{B_I} W^g_1 \right),
\end{equation}
where $\Phi(q)= \int_0^q y(t) \mathrm{d} x(t)$, $\alpha$ is the set of all poles of $y \mathrm{d} x$ and $W^g_1$, and $(A^I,B_I)$ is a canonical basis of cycles on $S$,
is only valid when the one-form $y \mathrm{d} x$ is residueless (\emph{i.e.} of the second kind), so that $\Phi(q)$ is meromorphic on $S$. Perhaps this is where the argument fails, since this identity is used in \cite{Eynard:2007ii}.

Note that the examples that we chose in the previous subsection to test symplectic invariance for the generalized recursion did not have this type of pathological behavior, hence for those cases symplectic invariance was expected to hold.

In any case, it would be very interesting to understand what assumptions need to be made on the spectral curves for the argument of \cite{Eynard:2007ii} to hold.  It would perhaps be even more interesting to obtain a direct and independent proof of symplectic invariance, more intrinsic to the topological recursion itself, clarifying the required assumptions along the way. Such a proof could then be extended to the generalized recursion proposed in this paper. Perhaps the definition of the free energies could also be slightly modified to make them symplectic invariant for all curves, including the pathological cases presented above. We are currently working on this and hope to report about it in the near future.

\section{Conclusion}
\label{s:conclusion}

In this paper we proposed a generalization of the Eynard-Orantin topological recursion for meromorphic functions $x$ and $y$ with arbitrary ramification. Our generalization is based on the work of Prats Ferrer \cite{Prats:2010}, where double ramification points were considered in the context of matrix models. We argued for the validity of our generalization by showing that in various examples it is consistent with invariance of the free energies $F_g$ under the highly non-trivial transformation $(x,y) \mapsto (y,x)$, and that the correlation functions satisfy important properties fulfilled by the corresponding objects in the Eynard-Orantin recursion. Along the way, we realized that symplectic invariance is in fact more subtle than expected, and much remains to be clarified.

A number of future avenues of research immediately come to mind. On the one hand, it would be interesting to obtain a matrix model derivation of the generalized recursion, in the spirit of \cite{Prats:2010}. Namely, one would start with a matrix model whose spectral curve has arbitrary ramification, and show that the generalized recursion is the solution of the loop equations. It should be possible to generalize the techniques of \cite{Prats:2010} in this context.
On the other hand, as we already discussed in subsections \ref{s:other} and \ref{s:SI}, it would be very interesting to study symplectic invariance of the generalized recursion further, and also prove various properties such as symmetry of the correlation functions, the dilaton equation, and the explicit formula for $W^0_3$. 

On a different note, the generalized recursion allows us to deal with functions $x$ and $y$ with arbitrary ramification. However, there is still a case that we cannot deal with: when $x$ has \emph{no} ramification point. \emph{A priori}, one would be tempted to say that in this case, all correlation functions and free energies are trivially zero, since at each step of the recursion there is no residue to be taken. But this does not seem to be the right approach in the context of symplectic invariance; one should rather say that the recursion is not well defined when $x$ has no ramification point. As an example, it is well known (see \cite{Norbury:2009,NS:2010}) that for the genus $0$ spectral curve with meromorphic functions $x = t + 1/t$ and $y=t$, the $F_g$'s are given by
\begin{equation}
F_g = \chi(\mathcal{M}_g) = \frac{B_{2g}}{2g (2g-2)},
\end{equation}
where $\chi(\mathcal{M}_g)$ is the orbifold Euler characteristic of the moduli space of genus $g$ curves, and $B_{2g}$ is the $2g$'th Bernoulli number. If symplectic invariance is to be believed, the free energies $\tilde{F}_g$ obtained from the spectral curve $(\tilde{x},\tilde{y})=(y,x)$ should be equal to the $F_g$ above, which are clearly non-zero, even though $\tilde{x} = y = t$ is a one-to-one map. Therefore, it is clear that to make sense of symplectic invariance, one needs to extend the recursion to cases where the $x$-map has no ramification. At the moment it is unclear to the authors how this may be done, or in fact if this can be done at all. It would be very interesting to investigate this further.

\end{document}